\documentclass[12pt,reqno]{amsart}

\usepackage{amsmath,amssymb,mathrsfs}
\usepackage{amsthm}

\newtheorem{assumption}{Assumption}
\usepackage{graphicx,cite,times}
\usepackage{cases}
\usepackage{color}
\usepackage{appendix}
\usepackage{enumerate}
\usepackage{bm} 
\usepackage{ulem}

\usepackage{mathtools}
\usepackage[ruled]{algorithm2e}
\usepackage{soul}
\usepackage{verbatim}

\usepackage[colorlinks, linkcolor=blue, citecolor=blue]{hyperref}

\usepackage{cite}
\graphicspath{{figures/}}
\usepackage{epstopdf}
\usepackage{rotating}
\usepackage{algorithmic}
\usepackage[pagewise]{lineno}

\newtheorem{theorem}{Theorem}[section]
\newtheorem{corollary}[theorem]{Corollary}
\newtheorem{lemma}[theorem]{Lemma}
\newtheorem{proposition}[theorem]{Proposition}
\newtheorem{remark}[theorem]{Remark}
\newtheorem{definition}[theorem]{Definition}

\numberwithin{equation}{section}

\begin{document}

\title[Vertex centrality reconstruction in an inverse problem for information diffusion]{Vertex centrality reconstruction in an inverse problem for information diffusion}

\author{Yixian Gao}
\address{School of Mathematics and Statistics, and Center for Mathematics and Interdisciplinary Sciences, Northeast Normal University, Changchun, Jilin 130024, China}
\email{gaoyx643@nenu.edu.cn}

\author{Songshuo Li}
\address{School of Mathematics and Statistics, and Center for Mathematics and Interdisciplinary Sciences, Northeast Normal University, Changchun, Jilin 130024, China}
\email{liss342@nenu.edu.cn}

\author{Yang Yang}
\address{Department of Computational Mathematics Science and Engineering, Michigan State University, East Lansing, MI 48824, USA}
\email{yangy5@msu.edu}

\thanks{The research of Y. Gao was supported by NSFC grant 12371187 and NSF of Jilin Province (Outstanding Youth Fund) 20240101006JJ. The research of Y. Yang is partially supported by the NSF grants DMS-2237534 and DMS-2220373.
}

\keywords{vertex centrality reconstruction, distributions of the first passage times, random walk, graphs}

\begin{abstract}
We consider an inverse problem in information diffusion modeled by random walks on combinatorial graphs. The problem concerns reconstruction of vertex centrality from the distribution of the first passage times observed on a subset of vertices. We adapt the boundary control method to obtain a direct algorithm that computes the unobserved vertex centrality. The algorithm is numerically implemented and validated on small graphs.
\end{abstract}

\maketitle

\section{introduction}

\subsection{Background}
Information diffusion studies spatiotemporal propagation of information in a social system. Here, ``information'' is interpreted in the broad sense such as infectious diseases, news, or ideas. Examples of information diffusion include the spread of infectious diseases through contact networks, the dissemination of news or opinions on social media platforms, and the propagation of signals or activities in biological or technological networks. In this paper, we study the mathematical modeling of information diffusion on a network from the perspective of inverse problems. We are interested in explicitly reconstructing certain network parameters from partial observation on a sub-network.

Let us begin by setting up the model. 
A network on which information propagates is modeled mathematically as a combinatorial graph, which consists of a set $X$ of vertices and a set $\mathcal{E}$ of edges. In information diffusion, $X$ represents the collection of individuals or agents, and $\mathcal{E}$ represents interactions between them. If two vertices $x,y\in X$ are connected by an edge, we write $x\sim y$, and represent the edge by $\{x, y\}$. In this paper, we focus on a special class of graphs. A graph is called \textit{finite} if $|X|$ and $|\mathcal{E}|$ are both finite, where $|\cdot|$ denotes the cardinality of a set. A graph is called \textit{undirected} if $\{x, y\} = \{y,x\}$ for all $x,y\in X$. A graph is called \textit{simple} if it has no self-loops and there is at most one edge between any pair of vertices. A graph is called \textit{connected} if any two vertices can be connected by a sequence of edges. Henceforth, all the graphs in this paper are assumed to be finite, undirected, simple and connected.

We assign some weights on the vertices and edges. The weight on vertices is defined by the function 
$$
\mu: X\to\mathbb{R}_{>0}
$$ 
where $\mathbb{R}_{>0}$ denotes the set of positive real numbers. For each $x\in X$, the value $\mu(x)$ is called the \textit{vertex centrality}, and characterizes the influence of the vertex $x$ in the network. 
The weight on the edges is defined by a function 
$$
w: X\times X\longrightarrow \mathbb{R}_{\ge 0},
$$ 
satisfying $w(x,y)=w(y,x)$, where $\mathbb{R}_{\ge 0}$ denotes the set of non-negative real numbers. For each pair of vertices $x,y\in X$, the value $w(x,y)$ is called the \textit{edge weight}, and characterizes the strength of the connection between $x$ and $y$. We require that $w(x,y)>0$ for all $x\sim y$, and $w(x,y) = 0$ for $x\nsim y$. 
Here, $\mathcal{E}=\big\{\{x,y\}\in X\times X~|~w(x,y)>0\big\}$. Throughout the paper, we will denote the graph with these weight functions by the quadruple $\mathbb{G}:= (X,\mathcal{E},\mu,w)$. For the brevity of notation, we often write $\mu(x)$ as $\mu_x$, and $w(x,y)$ as $w_{x,y}$.

The dynamics in information diffusion can be modeled in various ways. In this paper, we adopt a widely-used model where the propagation of information is regarded as a diffusion process (for example, see 
\cite{Datta_Dorlas_2004,Christley2005Infection,Draief2011RandomWalk}). To study the diffusion process on $\mathbb{G}$, we introduce the \textit{graph Laplacian} $\Delta$ whose action on a real-valued function $u:X \longrightarrow \mathbb{R}$ is defined by
\begin{equation*}
\Delta u(x):= \frac{1}{\mu(x)}\sum_{\substack{y\in X\\ y\sim x}}w(x,y)(u(y)-u(x)),\quad x\in X.
\end{equation*}
We remark that two types of graph Laplacian are of vital significance in graph theory:
the combinatorial Laplacian when $\mu \equiv 1$, and the normalized Laplacian when $w \equiv 1$ and $\mu(x) = {\rm deg} (x)$.
We define the discrete first-order time derivative of $u$ as
\begin{align*}
D_{t} u(t,x)= u(t+1,x)-u(t,x),\qquad t = 0,1,2,\dots\text{~and}  ~x\in X.
\end{align*}
If we observe the information diffusion only at discrete times, then its dynamics are governed by the following discrete-time graph heat equation: 
\begin{align*}
D_t u(t,x)=\Delta u(t,x),\quad (t,x)\in \mathbb{N}_0\times X,
\end{align*}
where $\mathbb{N}_0$ denotes the set of non-negative integers. Moreover, if the initial distribution of the information is described by a function $g:X\to\mathbb{R}$, then we can complement the graph heat equation with the initial condition 
$$
u(0, x) = g(x).
$$
The solution $u$ of this initial value problem describes the expected information level at time $t$ and vertex $x$.

Let us introduce notations for some inner products. For a subset of vertices $\Omega\subset X$ and functions $f,h: \Omega\longrightarrow \mathbb{R}$, we denote by $\ell^2(\Omega)$ the $\ell^2$-space of real-valued functions equipped with the inner product:
\begin{equation*}
\langle f,h\rangle_\Omega := \sum\limits_{x\in \Omega}\mu_x f(x) h(x).
\end{equation*}
An important fact, which will be used later, is that the graph Laplacian $\Delta$ is a self-adjoint operator on $\ell^2(X)$, that is, 
$$
(\Delta f,h)_X=(f,\Delta h)_X \qquad \text{ for all~} f,h: X\to \mathbb{R}.
$$
More generally, let us denote $\mathbb{Z}_T :=\{0,1,\dots, T-1\}$ and consider time dependent functions $u, v : \mathbb{Z}_T \times \Omega \to \mathbb{R}$, where $T\in \mathbb{N}_0\setminus \{0\}$. The inner product on the space $\ell^2(\mathbb{Z}_T \times \Omega)$ is defined by
\begin{equation*}
\langle u, v \rangle_{\mathbb{Z}_T \times \Omega} 
\coloneqq \sum_{t=0}^{T-1} \langle u(t), v(t) \rangle_\Omega
= \sum_{t=0}^{T-1} \sum_{x \in \Omega} \mu_x u(t,x) v(t,x).
\end{equation*}

\subsection{Problem Formulation}
It is well known that the discrete-time graph heat equation is closely connected to random walks on graphs. Here, we follow the exposition in~\cite{MR4598377} to describe the relation: Suppose the vertex centrality $\mu$ and edge weight $w$ satisfy $1\geq \frac{1}{\mu_x}\sum\limits_{\substack{y\in X\\ y\sim x}}w_{xy}$ for all $x\in X$. We define
\begin{align}\label{Heat_transition_proba_def}
p_{xy} := \frac{w_{xy}}{\mu_x} \quad\text{for}~ x\sim y, \quad \text{and} \quad p_{xx}:=1-\frac{1}{\mu_x}\sum\limits_{\substack{y\in X\\ y\sim x}}w_{xy}.
\end{align}
Then $\{p_{xy}: y\in X\}$ is a discrete probability distribution for each $x\in X$. Let $H_t^{x_0}$ be the discrete-time random walk with the state space $X$ and one-step transition probability $\{p_{xy}: x,y\in X\}$, which starts at $x_0\in X$ at time $t=0$, then
\begin{align*}
  p_{xy} = \mathbb{P}(H_{t+1}^{x_0}=y~|~H_{t}^{x_0}=x).
\end{align*} 
The random variable $H_t^{x_0}$ represents the vertex at which the information resides at time $t$. 

Moreover, for any function $g: X\rightarrow \mathbb{R}$, the expectation 
\begin{align}\label{u_def_by_expectation}
 u(t,x):=\mathbb{E}(g(H_t^x))=\sum\limits_{y\in X}\mathbb{P}(H_{t}^{x}=y)g(y)
\end{align}
satisfies the following initial value problem for the graph heat equation:
\begin{align}\label{graph_heat_eq}
\begin{cases}
D_t u(t,x) - \Delta u(t,x) = 0, & (t,x)\in \mathbb{N}_0\times X, \\
u(0,x) = g(x), & x\in X.
\end{cases}
\end{align}

Now, we can define the quantity that will be used as the measurement. Fix $x\in X$, and let $H_t^{x}$ be the random walk as above that starts at $x$. We define the \textit{first passage time} for $y\in X$ as
\begin{align*}
    \tau(x,y):=\inf~\{s\geq 1~|~H_{s}^{x}=y\},
\end{align*}
which is the minimal time it takes for the random walk that starts at $x$ to arrive at $y$. Note that $\tau(x,y)$ is a random variable defined on the set of all possible paths, taking values in $\mathbb{Z}_{>0} \cup \{\infty\}$. We define the \textit{distribution of the first passage times} as:
\begin{align*}
r(t,x,y):=\mathbb{P}(\{\tau(x,y)=t\}), \qquad t = 1,2,3,\dots.
\end{align*}

The inverse problem we are interested in is the following: Suppose information propagates in a network $\mathbb{G}$ as random walks, and the transition probabilities are structured as in~\eqref{Heat_transition_proba_def}, where $\mu$ is the vertex centrality, and $w$ is the known edge weight function.
Let $B\subset X$ be a fixed subset of vertices 
where observation can be made. 
If we release multiple information sources in $B$ and collect the distribution of first passage times of the information arriving at $B$, can we estimate the vertex centrality $\mu$? The problem can be made more precise using the notations we have introduced as follows:

\begin{center}
\textit{Given $(X,\mathcal{E}, \mu|_{B}, w)$ and the data $r|_{\mathbb{Z}_{2T}\setminus\{0\} \times B \times B}$ for some $T<\infty$, how can we recover $\mu|_{X\setminus B}$?
}\end{center}

Here, we assume $\mu|_{B}$ is known since $B$ is the set of observation, and  it remains to recover $\mu|_{X\setminus B}$ on the unobserved part of the network. The uniqueness of $\mu|_{X\setminus B}$ has been established (among other results) in \cite[Theorem 1.5]{MR4598377} under suitable geometric conditions. However, the proof is non-constructive, thus does not readily give a reconstruction method. In this paper, we focus on the reconstruction problem. Our central result is an explicit reconstruction formula for $\mu|_{X\setminus B}$ derived based on the boundary control method. The formula further leads to a direct, non-iterative algorithm (see Algorithm~\ref{alg:Framwork_of_Heateq}), which is numerically validated using simulations. Following \cite{MR4598377}, the major assumption for our derivation is:

\medskip
\begin{assumption}\label{Assum_eigenfunction}
Let $\mathbb{G}:= (X,\mathcal{E},\mu,w)$ be a finite, undirected, connected, simple graph, and $B\subset X$ a subset of vertices. Suppose
\begin{enumerate}
\item[(i)] there does not exist a nonzero eigenfunction of $-\Delta$ vanishing on $B$;
\item[(ii)] the edge weight satisfies
${\mu_x}\geq \sum\limits_{\substack{y\in X\\ y\sim x}}w_{xy}$ for all $x\in X$.
\end{enumerate} 
\end{assumption}

The condition (ii) is needed simply to ensure that the transition probabilities $p_{xx}$ defined in~\eqref{Heat_transition_proba_def} are non-negative. The condition (i) is a unique continuation principle on the graph. It ensures the exact controllability of the graph heat equation with controls localized in $B$ in~\cite[Proposition 5.3]{MR4598377}, which provides the foundation for our reconstruction method, see Lemma~\ref{span_space_l2X_heat}.

\subsection{Literature Review}
Most inverse problems concerning random walks and Markov chains on graphs focus on the recovery of transition probabilities under a prescribed graph topology. These include recursive reconstruction of Markov transition probabilities from boundary measurements \cite{MR1339870}, as well as the determination of transition probabilities from input–output travel time data on classes of directed graphs with loops \cite{MR2321831}. From a statistical and algorithmic perspective,  a convergent algorithm for recovering transition probabilities based on the relative net number of traversals along edges was proposed in \cite{MR3925423}. Related identifiability results have been established for various graph structures: in \cite{MR2455175}, it was shown that a simple Markov chain on the integers is uniquely determined—among chains differing from a reference chain only on a finite set—by the joint distribution of first hitting times and locations, while  an analogous uniqueness result for nondegenerate simple Markov chains on certain classes of finite rooted trees was proved in \cite{MR2479475}. In applied settings, partial observations of visiting frequencies and transition rates have been used to infer both initial states and transition probabilities of Markov chains, with applications to traffic monitoring systems \cite{mr105555}. The work that is most relevant to our paper is \cite{MR4598377}, which proves that, under suitable geometric conditions, the graph structure and the transition probabilities are uniquely determined by the distributions of the first passage times. These works demonstrate that random walk observables can encode sufficient information to recover transition mechanisms, forming the foundation for inverse problems on graphs driven by stochastic dynamics.

The reconstruction approach employed in this work is a variant of the boundary control method pioneered by Belishev~\cite{MR924687}. The boundary control method was originally developed for inverse problems associated with the continuous wave equation. Its central idea is that boundary excitations generate waves whose interior dynamics encode geometric and analytic information about the underlying domain and coefficients, and that this information can be accessed through suitably designed boundary controls and measurements. When combined with Tataru’s unique continuation results ~\cite{tataru95unique,Tataru99unique}, the boundary control method has proven to be a powerful tool for establishing identifiability of coefficients in evolution equations. In particular, a boundary control framework was developed for the continuous heat equation in~\cite{Avdonin1997BC} to recover the heat conductivity, which may be viewed as a continuous analogue of the approach pursued here, albeit with different types of measurements. A key advantage of the boundary control method lies in its constructive nature: it often yields explicit operator identities that relate measured boundary data, either directly or indirectly, to the unknown quantities of interest. Numerical implementations based on the boundary control method have been developed for continuous wave equations ~\cite{MR1694848,MR3483640,MR3825620,MR3959328,MR4844605,
pestov2012numerical,pestov2010numerical,MR4369044,MR4500893, yang2026formula,yang2026linearized} and for continuous heat equations~\cite{Avdonin1997BC}. We refer the reader to the surveys and monograph~\cite{belishev2017boundary,MR2353313,katchalov2001inverse} for comprehensive overviews of the boundary control method and its applications to continuous evolution equations.

The boundary control method has also been adapted to study inverse problems on graphs. The development includes determining planar metric trees and their edge densities from spectral data~\cite{MR2067494,MR2218385}, detecting cycles on metric graphs from boundary measurement~\cite{MR2545980}, computing vertex centrality on combinatorial graphs from spectral data~\cite{MR4907044}, and detecting network blockage~\cite{blaasten2019blockage}. The present work follows this line of research by developing a boundary control framework for the graph heat equation to address an inverse problem arising in the context of information diffusion.

\subsection{The Contribution} This paper's major contribution is an explicit reconstruction formula and a corresponding algorithm for computing the vertex centrality $\mu|_{X\setminus B}$ from the distribution of first passage times observed on a subset of vertices, see Algorithm~\ref{alg:Framwork_of_Heateq}. By adapting the boundary control method to the graph heat equation, we obtain a closed-form reconstruction formula, from which a direct, non-iterative algorithm is derived. The effectiveness of the proposed algorithm is demonstrated through numerical experiments. This approach has potential applications in identifying unknown graph structures in models of information diffusion.

Along the derivation of the reconstruction formula, we also obtain new uniqueness results as a byproduct. In particular, we establish uniqueness for the recovery of (the orthogonal projection of) the vertex centrality under weaker geometric assumptions; see Theorem~\ref{Heat_orthogonal_projection_of_mu_onto_Q} and Corollary~\ref{coro_dense_for_Heatweight}. Although the uniqueness of vertex centrality can be deduced from \cite[Theorem 1.5]{MR4598377}, the proof therein relies on a stronger geometric hypothesis known as the two-points condition. In contrast, our Corollary~\ref{coro_dense_for_Heatweight} establishes a generic uniqueness result under the sole requirement of the unique continuation principle stated in Assumption~\ref{Assum_eigenfunction}~(i). Since the two-points condition implies the unique continuation principle \cite[Proposition 2.6]{MR4598377}, our result demonstrates new cases where generic uniqueness holds under weaker assumptions.

\medskip
The paper is structured as follows: Section~\ref{sec:Blagoidentity} proves Blagovescenskii-type identities that link the data and several inner products on the graph. Section~\ref{Sec:Heat_Uniqueness and Reconstruction} utilizes the identities to derive the reconstruction procedure. Section~\ref{Sec:Heat_Reconstruction_algorithm} summarizes the reconstruction algorithm, and discusses details of implementation. Section~\ref{Sec:Heat_Numerical_experiment} presents numerical experiments on small graphs to validate and assess the algorithm.

\section{Blagovescenskii-type Identity} \label{sec:Blagoidentity}
This section establishes an explicit expression for the solution to the  nonhomogeneous heat equation \eqref{nonhomog_graph_heat_eq} following the approach of \cite{MR4598377}, and a Blagovescenskii-type Identity. The expression of the solution is given in terms of distributions of the first passage times
on the vertex observation set $B \subset X$ for discrete times $t \in \mathbb{Z}_{2T}\setminus\{0\}$.

\subsection{Representation of Graph Heat Solutions on \texorpdfstring{$B$}{}}\label{Sec:Distributions of first passage  times determine solution on observation set}

Let $U = U^f(t,x)$ denote the solution of the following nonhomogeneous heat equation on $X$ with the zero initial condition: 
\begin{align}\label{nonhomog_graph_heat_eq}
\begin{cases}
D_t U(t,x) - \Delta U(t,x) = f(t,x), & (t,x)\in \mathbb{Z}_{2T}\times X, \\
U(0,x) = 0, & x\in X.
\end{cases}
\end{align}
Let us identify $\ell^2(\mathbb{Z}_T\times B)$ with the subspace of $\ell^2(\mathbb{Z}_{2T}\times X)$ consisting of functions that are compactly supported on $\mathbb{Z}_T\times B$, that is,
$$
\ell^2(\mathbb{Z}_T\times B) \cong \{f\in \ell^2(\mathbb{Z}_{2T}\times X) \mid \operatorname{supp}(f)\subset \mathbb Z_{T}\times B \}.
$$
Specifically, we identify $f\in \ell^2(\mathbb{Z}_T\times B)$ with its zero extension in $\ell^2(\mathbb{Z}_{2T}\times X)$. We will denote such extension again by $f$. Henceforth, we will only consider sources $f\in\ell^2(\mathbb{Z}_{T}\times B)$ when solving for $U^f$.

The next lemma, proved in \cite[Lemma 3.4, Theorem 1.5]{MR4598377}, 
gives an explicit representation of $U^f$ in terms of $f$ and the data $r|_{\mathbb{Z}_{2T}\setminus\{0\} \times B \times B}$.

\begin{lemma}\label{Lemma_Uf_by_combinations_by_r}\cite[Lemma 3.4, Theorem 1.5]{MR4598377}
The solution $U^f|_{\mathbb{Z}_{2T}\times B}$ with $f\in \ell^2(\mathbb{Z}_T\times B)$ can be written in terms of $f$ and $r(t, x, y)|_{\mathbb{Z}_{2T}\setminus\{0\} \times B \times B}$ as follows:

\begin{align}\label{Uf_by_combinations_by_r}
U^f(t,x) =
\begin{cases}
0, \quad & t=0, \\
f(0,x),\quad &t=1,\\
\sum\limits_{y\in B}f(t-2,y)\cdot r(1,x,y)+\sum\limits_{s=1}^{t-2}\sum\limits_{y\in B} \bigg\{f(s-1,y)\Big( r(t-s,x,y)\\
\quad +\sum\limits_{j=2}^{t-s}\sum\limits_{1\leq t_1< t_2<\cdots<t_j=t-s} r(t_1,x,y)\cdot  \prod\limits_{i=2}^jr(t_i-t_{i-1},y,y) \Big) \bigg\} + f(t-1,x),\quad & t\geq 2
\end{cases}
\end{align}
for $(t,x)\in \mathbb Z_{2T}\times B$. 

\end{lemma}

\begin{proof}
We begin by noting that the initial condition for the heat equation \eqref{nonhomog_graph_heat_eq} is $U^f(0,x) = 0$ for all $x \in X$.
Following the approach in \cite[Lemma 3.4]{MR4598377}, the solution $U^f$ can be expressed as a linear combination of fundamental solutions $u_y : \mathbb{Z}_{2T} \times X \to \mathbb{R}$ to the homogeneous heat equation with Dirac initial conditions
\begin{align}\label{fundamental_homheat}
\begin{cases}
D_t u_y(t,x) - \Delta u_y(t,x) = 0, & (t,x)\in \mathbb Z_{2T}\times X, \\
u_y(0,x) = \delta_{y}(x), & x\in X
\end{cases}
\end{align}
for each fixed $y\in B$. Here, $\delta_y(x) = 1$ if $x=y$, and $\delta_y(x) = 0$ if $x\neq y$. Specifically, for $s \in \mathbb{Z}_+$,  we have
\begin{align}\label{Uf_by_combinations}
 U^f(t,x)=
\sum\limits_{s=1}^t\sum_{y\in B} f(s-1,y)u_y(t-s,x),\quad (t,x)\in \mathbb{Z}_{2T}\setminus\{0\}\times X,
\end{align}
with $U^f(0,x)=0$ for $x\in X$.

Based on the definition of $u$ in \eqref{u_def_by_expectation}, the fundamental solutions $u_y$ admit a probabilistic representation in terms of the first passage time distributions $r$, which is established in the proof of \cite[Theorem 1.5]{MR4598377}:
\begin{align}\label{delta_by_rtxy}
u_y(t,x) & =\mathbb{E}(\delta_y(H_t^{x}))=\mathbb{P}(H_t^{x}=y) \nonumber\\
& =\begin{cases}
r(1,x,y),\quad t=1,\\
r(t,x,y)+\sum\limits_{j=2}^t\sum\limits_{1\leq t_1< t_2<\cdots<t_j=t}r(t_1,x,y)\cdot \prod\limits_{i=2}^jr(t_i-t_{i-1},y,y),\quad t\in\{2,\dots,2T-1\}
\end{cases}
\end{align}
for all $x\in X$ and $y\in B$. 

The expression \eqref{delta_by_rtxy} decomposes the event $\{H_t^x = y\}$ according to $j$ distinct visiting times to the vertex $y$. The term $r(t,x,y)$ corresponds to the walk first reaching $y$ exactly at time $t$. For $j \geq 2$, the inner sum accounts for sequences where the walk visits $y$ at exactly $j$ distinct times $t_1 < \cdots < t_j = t$, with $r(t_1,x,y)$ representing the probability of the first arrival at $y$ at time $t_1$ from $x$, and each factor $r(t_i - t_{i-1}, y, y)$ giving the probability of the first return to $y$ (exactly at time $t_i$) from $y$ at time $t_{i-1}$ with $2\leq i\leq j$.

Substituting the  equality \eqref{delta_by_rtxy} into the equality  \eqref{Uf_by_combinations} for $t\in\mathbb{Z}_{2T}\setminus\{0\}$, and observing that  
\begin{align*}
\sum_{y\in B} f(t-1,y)\delta_y(x)=f(t-1,x)
\end{align*}
since $f$ is supported on $\mathbb{Z}_T\times B$, we obtain \eqref{Uf_by_combinations_by_r}.

\end{proof}

\begin{remark}\label{Operation_complexity_of_Uf}
The nested summations in \eqref{Uf_by_combinations_by_r} make the computational cost grow exponentially in $T$. To see this, let us count the number of basic arithmetic operations (i.e, addition, subtraction, multiplication and division). We will ignore the computational cost to evaluate functions.
Fix \(t\) and \(x\in B\). The dominant cost in \eqref{Uf_by_combinations_by_r} comes from the combinatorial block
\[
\sum_{s=1}^{t-1}\sum_{y\in B}
\sum_{j=2}^{t-s}\sum_{1\le t_1<\cdots<t_j=t-s}
r(t_1,x,y)\prod_{i=2}^j r(t_i-t_{i-1},y,y).
\]
For \(m:=t-s\), there are \(2^{m-1}-1\) summands, and the sum of the multiplication counts over \(j\) equals
\(\sum_{j=2}^{m} (j-1)\binom{m-1}{j-1}=(m-1)2^{m-2}\). Hence, for fixed \(t,x\), the number of multiplications contributed by the combinatorial block is
\[
|B|\sum_{m=1}^{t-1}(m-1)2^{m-2}=|B|\big((t-3)2^{\,t-2}+1\big).
\]
Summing over \(x\in B\) and \(t\in\mathbb{Z}_{2T}\setminus\{0\}\), the total arithmetic complexity is
\[
\sum^{2T-1}_{t=1} |B|^2\big((t-3)2^{\,t-2}+1 \big) = |B|^2\big(2T+1+(2T-5)2^{2T-2}\big).
\]
The analysis shows that the computational cost of the nested summations in \eqref{Uf_by_combinations_by_r} is at least on the order of $O(|B|^2 T 2^{2T})$. Since Lemma \ref{span_space_l2X_heat} below requires $T \geq |X|$, the computational complexity scales in $|X|$ at least as $O(|B|^2 |X| 2^{2|X|})$. This exponential growth in $|X|$ becomes a crucial limiting factor for the size of the graphs in the numerical experiments in Section \ref{Sec:Heat_Numerical_experiment}.
\end{remark}

From the proof of Lemma \ref{Lemma_Uf_by_combinations_by_r}, we observe that for $t\in\mathbb{Z}_{2T}\setminus\{0\}$ and $x,y\in B$, 
both the fundamental  solution $u_y(t,x)$ to the equation \eqref{fundamental_homheat}, and the solution $U^f(t,x)$ to the nonhomogeneous heat equation \eqref{nonhomog_graph_heat_eq} can  be determined by $r(t, x, y)|_{\mathbb{Z}_{2T}\setminus\{0\} \times B \times B}$.

\subsection{Blagovescenskii-type Identity}

Recall that $U^f$ denotes the solution of the non-homogeneous heat equation~\eqref{nonhomog_graph_heat_eq} with $f\in\ell^2(\mathbb Z_{T}\times B)$. In particular, $U^f|_{\mathbb{Z}_{2T}\times B}$ can be obtained from $f$ and the data $r|_{\mathbb{Z}_{2T}\setminus\{0\} \times B \times B}$ by Lemma \ref{Lemma_Uf_by_combinations_by_r}.

We introduce a few linear operators that will be used later. First, define the \textit{source-to-solution map}
\begin{equation}\label{Map_of_source to solution_of_Nheat}
\begin{aligned}
 \Lambda_{\mu}: \ell^2(\mathbb{Z}_T\times B) & \to \ell^2(\mathbb{Z}_T\times B) \\
f & \mapsto U^f|_{\mathbb{Z}_T\times B}.
\end{aligned}
\end{equation}
Note that by Lemma \ref{Lemma_Uf_by_combinations_by_r}, the operator $\Lambda_\mu$ can be explicitly computed using the data $r|_{\mathbb{Z}_{2T}\setminus\{0\} \times B \times B}$.

Next, define
\begin{align*}
    W: \ell^2(\mathbb{Z}_T\times B) &\to \ell^2(X), \\
    f &\mapsto U^f(T,\cdot).
\end{align*}
$W$ is clearly a linear operator. Its $\ell^2$-adjoint is denoted by $W^*$. The following result, established in \cite{MR4598377}, shows that $W$ is surjective when $T$ is sufficiently large.

\begin{lemma}\label{span_space_l2X_heat}\cite[Proposition 5.3]{MR4598377} 
Suppose the graph satisfies Assumption~\ref{Assum_eigenfunction} and $T \geq |X|$. Then
    \begin{align*}
        \left\{ U^f(T, \cdot) \mid f\in\ell^2(\mathbb{Z}_{T}\times B) \right\} = \ell^2(X).
    \end{align*}

\end{lemma}

For a function $u(t,x)$, we introduce the time reversal operator: 
\begin{align*}
R_{T-1}: \ell^2(\mathbb{Z}_T \times B) &\to \ell^2(\mathbb{Z}_T \times B), \\
 u(t, \cdot) &\mapsto u(T-1-t, \cdot), \quad t \in \mathbb{Z}_T.
\end{align*}
The operator $R_{2T-1}$ 
is defined likewise: 

\begin{align*}
R_{2T-1}: \ell^2(\mathbb{Z}_{2T} \times B) &\to \ell^2(\mathbb{Z}_{2T} \times B), \\
 u(t, \cdot) &\mapsto u(2T-1-t, \cdot), \quad t \in \mathbb{Z}_{2T}.
\end{align*}

We also define the temporal projection operator on $B$: 
\begin{align*}
P_T: \ell^2(\mathbb{Z}_{2T} \times B) &\to \ell^2(\mathbb{Z}_T \times B), \\
 u &\mapsto u|_{\mathbb{Z}_{T}\times B}.
\end{align*}

Let $U^{f_1}$ and $U^{f_2}$ denote the solutions of the heat equation \eqref{nonhomog_graph_heat_eq} with sources $f_1, f_2\in\ell^2(\mathbb{Z}_T\times B)$, respectively.
The next proposition shows that the inner product of these solutions on $X$ at time $t=T$ can be represented by $f_1$ and  $f_2$. 
This type of result, known as the Blagovescenskii identity~\cite{Blagoveshchenskii1967}, plays a crucial role in the framework of the boundary control method. Henceforgth, for a function $u(t,x)$, we write $u(t)$ for the spatial function $u(t, \cdot)$.

\begin{proposition}\label{prop:product_of_Uf12}
For any $f_1, f_2\in\ell^2(\mathbb{Z}_T\times B)$, we have 

\begin{equation}\label{Heat_Inner_product_of_U^f}
\langle U^{f_1}(T),U^{f_2}(T)\rangle_X =\langle W{f_1},W{f_2}\rangle_X  = \langle f_1, P_TR_{2T-1} U^{f_2} \rangle_{\mathbb{Z}_T\times B}.
\end{equation}
In particular, $W^*W f_2=P_TR_{2T-1} U^{f_2}$ for any $f_2\in\ell^2(\mathbb{Z}_T\times B)$.

\end{proposition}

\begin{proof}
The first equality follows from the definition of $W$. It remains to prove the second equality.

Define $I(s,t) := \langle U^{f_1}(s), U^{f_2}(t) \rangle_X$. 
Direct computation yields
\begin{align*}
(D_s - D_t) I(s,t) 
&= \langle D_s U^{f_1}(s), U^{f_2}(t) \rangle_X - \langle U^{f_1}(s), D_t U^{f_2}(t) \rangle_X \\
&= \langle \Delta U^{f_1}(s) + f_1(s), U^{f_2}(t) \rangle_X - \langle U^{f_1}(s), \Delta U^{f_2}(t) + f_2(t) \rangle_X \\
&= \langle \Delta U^{f_1}(s), U^{f_2}(t) \rangle_X + \langle f_1(s), U^{f_2}(t) \rangle_X \\
&\quad - \langle U^{f_1}(s), \Delta U^{f_2}(t) \rangle_X - \langle U^{f_1}(s), f_2(t) \rangle_X.
\end{align*}
By the self-adjointness of the operator $\Delta$ with respect to $\langle \cdot, \cdot \rangle_X$, 
we have $\langle \Delta U^{f_1}(s), U^{f_2}(t) \rangle_X = \langle U^{f_1}(s), \Delta U^{f_2}(t) \rangle_X$. Therefore, the terms involving $\Delta$ cancel, and we obtain:
\begin{align*}
(D_s - D_t) I(s,t) & = \langle f_1(s), U^{f_2}(t) \rangle_X - \langle U^{f_1}(s), f_2(t) \rangle_X \\
 & = \langle f_1(s), U^{f_2}(t) \rangle_B - \langle U^{f_1}(s), f_2(t) \rangle_B, \quad 0 \leq s, t \leq 2T-1
\end{align*}
where the second inequality holds since $\operatorname{supp}(f_1), \operatorname{supp}(f_2) \subset \mathbb{Z}_T \times B$.

Let us denote the right hand side as $F(s,t) := \langle f_1(s), U^{f_2}(t) \rangle_B - \langle U^{f_1}(s), f_2(t) \rangle_B$. 
From the definitions of $D_s$ and $D_t$, we obtain the recurrence relation
\begin{align*}
I(s+1,t)- I(s,t) -(I(s,t+1)- I(s,t)) = F(s,t), \quad 0\leq s,~t\leq 2T-1.
\end{align*}
In other words,
\begin{align}\label{Uf_induction_of_innerproduct}
I(s+1,t)=I(s,t+1)+F(s,t), \quad 0\leq s,~t\leq 2T-1.
\end{align}
We now prove by induction that for $1 \leq s \leq 2T-1$ and $0 \leq t \leq 2T-s$,
\begin{equation}\label{HeatUU_induction_equality}
I(s,t) = \sum_{j=0}^{s-1} F(j, t+s-1-j).
\end{equation}
By the definition of $I(s,t)$, we have
\begin{align*}
I(0,t) &= \langle U^{f_1}(0), U^{f_2}(t) \rangle_X = 0, \quad t\geq 0.
\end{align*}
For the base case $s=1$ in \eqref{HeatUU_induction_equality}, it follows from  equality \eqref{Uf_induction_of_innerproduct} that
\begin{align*}
I(1,t) = I(0,t+1) + F(0,t) = F(0,t), \quad 0 \leq t \leq 2T-1,
\end{align*}
which agrees with \eqref{HeatUU_induction_equality}.
Now we assume that \eqref{HeatUU_induction_equality} holds for $s \leq  k$ and prove the case $s = k+1$.
Indeed, using the recursive relation~\eqref{Uf_induction_of_innerproduct}, we obtain
\begin{align*}
I(k+1,t) &= I(k,t+1) + F(k,t) \\
&= \sum_{j=0}^{k-1} F(j, t+1+k-1-j) + F(k,t) \\
&= \sum_{j=0}^{k} F(j, t+k-j), \quad 0 \leq t \leq 2T-k-1,
\end{align*}
which justifies the case $s = k+1$.

Substituting the definitions of $I(s,t)$ and $F(s,t)$ into \eqref{HeatUU_induction_equality} yields
\[
\langle U^{f_1}(s), U^{f_2}(t) \rangle_X = \sum_{j=0}^{s-1} \left[ \langle f_1(j), U^{f_2}(t+s-1-j) \rangle_B - \langle U^{f_1}(j), f_2(t+s-1-j) \rangle_B \right]
\]
for $1 \leq s \leq 2T-1$ and $0 \leq t \leq 2T-s$.
Finally, taking $s = t = T$ and using $\operatorname{supp}(f_2) \subset \mathbb{Z}_T \times B$, we obtain
\[
\langle U^{f_1}(T), U^{f_2}(T) \rangle_X = \sum_{j=0}^{T-1} \langle f_1(j), U^{f_2}(2T-1-j) \rangle_B 
= \langle f_1, P_TR_{2T-1}U^{f_2} \rangle_{\mathbb{Z}_T \times B}.
\]
This proves the second equality in~\eqref{Heat_Inner_product_of_U^f}. The representation of $W^*W$ follows from the adjoint relation $\langle W f_1, W f_2 \rangle_X = \langle f_1, W^* W f_2 \rangle_{\mathbb{Z}_T \times B}$ and the arbitrariness of $f_1,f_2\in \ell^2(\mathbb{Z}_T\times B)$.
\end{proof}

Next, we introduce a class of functions defined on $X$ that will be used.

\begin{definition}
A function $\varphi: X \to \mathbb{R}$ is said to be \textit{harmonic} on $X \setminus B$, if 
\[
\Delta \varphi(x) = 0 \quad \text{for all } x \in X \setminus B.
\]
\end{definition}

\begin{remark}\label{Remark_harmonic_undependmu}
Given a function $\varphi: X \to \mathbb{R}$, we can determine whether it is harmonic on $X \setminus B$ without knowing $\mu|_{X \setminus B}$. This is because $\Delta \varphi(x) = 0$ if and only if 
$$
0 = \mu_x \Delta \varphi(x) = \sum_{\substack{y\in X\\ y\sim x}}w_{x,y}(u(y)-u(x)),
$$
where the right-hand side does not involve $\mu$.
\end{remark}

Let us introduce another spatial projection operator
$$
P_B: \ell^2(X) \to \ell^2(B), \qquad \varphi \mapsto \varphi|_{B}.
$$
The next proposition shows that the inner product between $U^{f}(T)$ and a function $\varphi$ that is harmonic on $X \setminus B$ can be computed from the source-to-solution map $\Lambda_\mu$.

\begin{proposition}\label{prop:product_of_Uf1phi}
Let $\varphi$ be harmonic on $X \setminus B$, and suppose the graph satisfies Assumption~\ref{Assum_eigenfunction}. Then for any $f\in\ell^2(\mathbb{Z}_T\times B)$, we have 

\begin{equation} \label{eq: prodwithharmonic}
\langle U^{f}(T), \varphi \rangle_X = \langle W f, \varphi \rangle_X = \langle f, R_{T-1} \Lambda_{\mu} R_{T-1} P_B \Delta \varphi + P_B \varphi \rangle_{\mathbb{Z}_T \times B}.
\end{equation}
Consequently, 

\begin{equation} \label{Heat_Wstar_formula}
W^* \varphi = R_{T-1} \Lambda_{\mu} R_{T-1} P_B \Delta \varphi + P_B \varphi\quad \in ~\ell^2(\mathbb{Z}_T \times B).
\end{equation}
\end{proposition}

\begin{proof}
The first equality in~\eqref{eq: prodwithharmonic} follows from the definition of $W$. To prove the second inequality in~\eqref{eq: prodwithharmonic}, define $J(t) := \langle U^{f}(t), \varphi \rangle_X$. Then
\begin{align*}
D_t J(t) &= \langle D_t U^{f}(t), \varphi \rangle_X \\
&= \langle \Delta U^{f}(t) + f(t), \varphi \rangle_X \\
&= \langle U^{f}(t), \Delta \varphi \rangle_X + \langle f(t), \varphi \rangle_X \\
&= \langle U^{f}(t), \Delta \varphi \rangle_{(X \setminus B)\cup B}  + \langle f(t), \varphi \rangle_X \\
&= \langle U^{f}(t), P_B \Delta \varphi \rangle_B + \langle f(t), P_B \varphi \rangle_B, \qquad t\in\mathbb{Z}_{2T}.
\end{align*} 
Here, the second equality holds since $U^f$ solves~\eqref{nonhomog_graph_heat_eq}, the third equality holds since $\Delta$ is self-adjoint with respect to the inner product on $X$, and the last equality holds since $\Delta \varphi=0$ on $X \setminus B$ and $\operatorname{supp}(f) \subset \mathbb{Z}_T \times B$.
As $J(0) = \langle U^{f}(0), \varphi \rangle_X = 0$, we obtain the following recurrence relation from the definition of $D_t$:
\[
J(t+1) - J(t) = \langle U^{f}(t), P_B \Delta \varphi \rangle_B + \langle f(t), P_B \varphi \rangle_B, \qquad t\in\mathbb{Z}_{2T}.
\]
Therefore, we have
\begin{align*}
J(T) &= J(0) + \sum_{j=0}^{T-1} [J(j+1) - J(j)] \\
&= \sum_{j=0}^{T-1} \left[ \langle U^{f}(j), P_B\Delta \varphi \rangle_B + \langle f(j), P_B\varphi \rangle_B \right] \\
&= \langle U^{f}, P_B \Delta \varphi \rangle_{\mathbb{Z}_T \times B} + \langle f, P_B \varphi \rangle_{\mathbb{Z}_T \times B} \\
&= \langle \Lambda_{\mu} f, P_B \Delta \varphi \rangle_{\mathbb{Z}_T \times B} + \langle f, P_B \varphi \rangle_{\mathbb{Z}_T \times B} \\
&= \langle f, R_{T-1} \Lambda_{\mu} R_{T-1} P_B \Delta \varphi \rangle_{\mathbb{Z}_T \times B} + \langle f, P_B \varphi \rangle_{\mathbb{Z}_T \times B} \\
&= \langle f, R_{T-1} \Lambda_{\mu} R_{T-1} P_B \Delta \varphi + P_B \varphi \rangle_{\mathbb{Z}_T \times B}.
\end{align*}
Here, the fourth equality uses the definition of the source-to-solution map $\Lambda_{\mu}$, and the fifth employs the adjoint property $\Lambda_{\mu}^* = R_{T-1} \Lambda_{\mu} R_{T-1}$, which is proved in Appendix~\ref{Sec:aAppSourceSolutionAdjiont}. This proves the second inequality in~\eqref{eq: prodwithharmonic}.
Finally, the equality~\eqref{Heat_Wstar_formula} follows from the relation $\langle W f, \varphi \rangle_X = \langle f, W^* \varphi \rangle_{\mathbb{Z}_T \times B}$. 
\end{proof}

Although Proposition~\ref{prop:product_of_Uf1phi} holds for any $\varphi$ that is harmonic on $X\setminus B$, it suffices to consider just a few basis functions. Indeed, if we index the vertices in $X \setminus B$ as $x_1, \dots, x_{|X \setminus B|}$ and those in $B$ as $x_{|X \setminus B|+1}, \dots, x_{|X|}$, then the unique solution of the boundary value problem (uniqueness guaranteed by \cite[Lemma 10.1]{MR4907044})
\[
    \Delta u(x) = 0, \quad x\in X\setminus B, \qquad   u|_{B} = g
\]
is a linear combination of $\varphi^{(j)}$ that satisfies
\begin{equation}\label{Dirichlet_problem}
    \Delta \varphi^{(j)}(x) = 0 ~\text{ for }~ x\in X\setminus B, \qquad \varphi^{(j)}|_{B} = \delta^{(j)}.
\end{equation}
Here, $\delta^{(j)}$ denotes the function on $B$ defined by
\[
\delta^{(j)}(x) := 
\begin{cases}
1 & \text{if } x = x_{|X \setminus B|+j}, \\
0 & \text{if } x \in B \setminus \{x_{|X \setminus B|+j}\}.
\end{cases}
\]

\section{ Reconstruction Procedure}\label{Sec:Heat_Uniqueness and Reconstruction}

In this section, we derive the procedure to reconstruct the unknown vertex centrality $\mu|_{X \setminus B}$.

\subsection{Construction of controls}

We begin by constructing suitable control functions on $B$. Here, a control function $h_0$ is a function in $\ell^2(\mathbb{Z}_T\times B)$ such that $U^{h_0}(T) := U^{h_0}(T,\cdot)$ is a prescribed spatial function.

\begin{proposition}\label{Heat_B_harmonic_explicit h0}
Suppose the graph satisfy Assumption~\ref{Assum_eigenfunction} and $T \geq |X|$. For any $\psi$ that is harmonic on $X \setminus B$, the function 
\begin{equation}\label{Heat_B_h0formula}
h_0 := (W^*W)^\dagger W^* \psi \quad \in \ell^2(\mathbb{Z}_T \times B)
\end{equation}
satisfies $U^{h_0}(T, x) = \psi(x)$ for all $x \in X$, where $(\cdot)^\dagger$ denotes the pseudo-inverse. 
\end{proposition}

\begin{proof}
By Lemma~\ref{span_space_l2X_heat}, the operator $W$ is surjective,  so the linear system $U^{h}(T)=Wh=\psi$ admits solutions. The solution $h_0$ given in~\eqref{Heat_B_h0formula} is the minimum norm solution.

\end{proof}

Note that by Proposition~\ref{prop:product_of_Uf12} and Proposition~\ref{prop:product_of_Uf1phi}, both $W^*W$ and $W^*\psi$ can be computed from $\Lambda_\mu$. As the latter is computable from the data $r|_{\mathbb{Z}_{2T}\setminus\{0\} \times B \times B}$ by Lemma~\ref{Lemma_Uf_by_combinations_by_r}, we see that $h_0$ can be explicitly computed from the data without knowing $\mu|_{X\setminus B}$.

\subsection{Reconstruction of the vertex centrality on \texorpdfstring{$X\setminus B$}.}

The rest of the approach closely follows the idea in our earlier work \cite{MR4907044}, so we just outline the idea here.

Define the space spanned by products of functions that are harmonic on $X\setminus B$:
\begin{align*}
Q :=\operatorname{span}\{(\varphi\psi)|_{X\setminus B}~:~\varphi,\psi\in \ell^2(X),~\Delta \varphi(x)=\Delta\psi(x)=0,~x\in X\setminus B\}. 
\end{align*}
Note that $Q$ is a subspace of $\ell^2(X \setminus B)$, and the definition of $Q$ is independent of $\mu_{X \setminus B}$ since the functions that are harmonic on $X \setminus B$ are independent of $\mu|_{X \setminus B}$, see Remark~\ref{Remark_harmonic_undependmu}.
Given two functions $\varphi,\psi$ that are harmonic on $X\setminus B$, we have
\begin{align*}
\sum_{x\in B\cup (X\setminus B)} \mu_x \psi(x) \varphi(x) 
&= \langle \psi, \varphi \rangle_X 
= \langle U^{h_0}(T), \varphi \rangle_X 
= \langle W h_0, \varphi \rangle_X 
= \langle h_0, W^* \varphi \rangle_{\mathbb{Z}_T \times B},
\end{align*}
where $h_0$ is chosen as in Proposition \ref{Heat_B_harmonic_explicit h0}. This implies
\begin{align}\label{Heat_miureconstruction_formula}
\sum_{x \in X \setminus B} \mu_x \psi(x) \varphi(x)
= \langle h_0, W^* \varphi \rangle_{\mathbb{Z}_T \times B}
- \sum_{x \in B} \mu_x \psi(x) \varphi(x).
\end{align}
As the right hand side can be computed from the data $r|_{\mathbb{Z}_{2T}\setminus\{0\} \times B \times B}$ and the known $\mu|_{B}$, we can take $\varphi,\psi = \varphi^{(j)}$ with $j=1,2,\dots, |B|$, respectively, to obtain a linear system, from which $\mu|_{X\setminus B}$ can be solved. In summary, we have proved:

\begin{theorem}\label{Heat_orthogonal_projection_of_mu_onto_Q}
Let $\mathbb{G} = (X,\mathcal{E},\mu,w)$ satisfy Assumption~\ref{Assum_eigenfunction}. 
Suppose $(X,\mathcal{E},\mu|_{B},w)$ are given, and $T \geq |X|$. Then the orthogonal projection of $\mu|_{X \setminus B}$ onto the subspace $Q \subset \ell^2(X \setminus B)$ can be explicitly reconstructed from the data $r|_{\mathbb{Z}_{2T}\setminus\{0\} \times B \times B}$. 
\end{theorem}

Theorem~\ref{Heat_orthogonal_projection_of_mu_onto_Q} only reconstructs the orthogonal projection of $\mu|_{X \setminus B}$ onto $Q$. To recover the full vertex centrality $\mu|_{X \setminus B}$, we need additional conditions on $X \setminus B$ and the edge weight function $w$. Let us again index the vertices in $X \setminus B$ as $x_1, \dots, x_{|X \setminus B|}$ and those in $B$ as $x_{|X \setminus B|+1}, \dots, x_{|X|}$.
Using the indexing, a function $g \in \ell^2(X)$ can be identified with a vector $\vec{g} = (g(x_1), \dots, g(x_{|X|}))^\top \in \mathbb{R}^{|X|}$ via the isomorphism
\begin{equation} \label{eq:Heat_identity}
\ell^2(X) \ni g \leftrightarrow \vec{g} := 
\begin{pmatrix}
\vec{g}_{X \setminus B} \\
\vec{g}_{B}
\end{pmatrix}
\in \mathbb{R}^{|X \setminus B|} \times \mathbb{R}^{|B|},
\end{equation}
where $(\cdot)^\top$ denotes transpose.
Let $\vec{\varphi}^{(j)}$ denote the vector representation of the unique solution to the boundary value problem \eqref{Dirichlet_problem}.
We define the space of functions that are harmonic on $X \setminus B$ in the vector form as:
\begin{equation}\label{Heat_harmonic_span}
\vec{H} := \operatorname{span} \left\{ \vec{\varphi}^{(j)} \in \mathbb{R}^{|X|} : j = 1, 2, \dots, |B| \right\}.
\end{equation}
Construct a matrix 
\[ \mathbf{H} \in \mathbb{R}^{|X\setminus B|\times \frac{|B|(|B|+1)}{2}} \] 
whose columns are the vectors
\( \vec{\varphi}|_{X\setminus B}^{(j)}\odot \vec{\varphi}|_{X\setminus B}^{(k)} \) with \( 1\le j\le k\le |B| \), where \(\odot\) denotes the Hadamard product. 
By Lemma 5.2 and Remark 5.3 in \cite{MR4907044}, we have 
\begin{equation} \label{eq:QequivH}
Q=\ell^2(X\setminus B)  \quad \text{ if an only if } \quad \operatorname{rank}({\bf{H}})=|X\setminus B|.
\end{equation}
Note that a necessary condition for $\operatorname{rank}({\bf{H}})=|X\setminus B|$ is $\frac{|B|(|B|+1)}{2} \ge |X \setminus B|$, which requires the set $B$ to contain sufficient vertices.

Using the vertex ordering, we can list the edges lexicographically and view the edge weight function $w$ as a vector with positive entries, that is $w\in \mathbb{R}_+^{|\mathcal{E}|}$. Since the Laplacian operator is defined in terms of $w$, the matrix $\mathbf{H}=\mathbf{H} (w)$ can be regarded as a matrix-valued function of $w\in \mathbb{R}^{|\mathcal{E}|}_+$. Applying \cite[Proposition 5.4]{MR4907044}, we obtain the following reconstruction for a generic set of $\omega$.

\begin{corollary} \label{coro_dense_for_Heatweight}
Let $\mathbb{G}=(X,\mathcal{E}, w,\mu)$ satisfy Assumption \ref{Assum_eigenfunction}. 
Suppose $(X,\mathcal{E},\mu|_{B},w)$ are given, $T\ge |X|$ and $\frac{|B|(|B|+1)}{2}\geq |X\setminus B|$.
If $Q=\ell^2(X\setminus B)$ holds for at least one edge weight $w$, then it holds for all edge weights except a set of measure zero in $\mathbb{R}_+^{|\mathcal{E}|}$. As a result, Theorem~\ref{Heat_orthogonal_projection_of_mu_onto_Q} reconstructs $\mu|_{X\setminus B}$ for all edge weights except a set of measure zero in $\mathbb{R}_+^{|\mathcal{E}|}$.
\end{corollary}

\begin{proof}
Let $\beta$ denote an arbitrary selection of $|X\setminus B|$ columns from $\mathbf{H}$. Define
$$
S_\beta := \{w\in \mathbb{R}_+^{|\mathcal{E}|} : \det(\mathbf{H}_{:,\beta})=0,\quad \text{and } ~{\mu_x}\geq \sum\limits_{\substack{y\in X\\ y\sim x}}w_{xy}\}.
$$
As $\det(\mathbf{H}_{:,\beta})$ is a rational function of $w$, its zero set $S_\beta$ is either the entire $\mathbb{R}_+^{|\mathcal{E}|}$ or a set of measure zero. The collection of $w$ that ensures ${\rm rank}(\mathbf{H}) < |X\setminus B|$ is
\begin{align*}
     &\quad \{w\in \mathbb{R}_+^{|\mathcal{E}|} : {\rm rank}(\mathbf{H}) < |X\setminus B|,\quad \text{and } ~{\mu_x}\geq \sum\limits_{\substack{y\in X\\ y\sim x}}w_{xy}\} \\
      & = \{w\in \mathbb{R}_+^{|\mathcal{E}|} : \det(\mathbf{H}_{:,\beta})=0, \; \forall\beta,\quad \text{and } ~{\mu_x}\geq \sum\limits_{\substack{y\in X\\ y\sim x}}w_{xy}\} \\
       & = \bigcap_\beta S_\beta,
\end{align*}
which is a finite intersection of $S_\beta$. If $Q=\ell^2(X\setminus B)$ for at least one $w$, then $\operatorname{rank}({\bf{H}})=|X\setminus B|$ for such $w$ by~\eqref{eq:QequivH}, hence the intersection cannot be the entire space, which must be a set of measure zero. When $Q=\ell^2(X\setminus B)$, the orthogonal projection of $\mu|_{X\setminus B}$ onto $Q$ is just $\mu|_{X\setminus B}$ itself.
\end{proof}

Numerical reconstruction of $\mu|_{X\setminus B}$ is summarized in Algorithm~\ref{alg:Framwork_of_Heateq} in Section~\ref{Sec:Heat_Reconstruction_algorithm}.

\section{Reconstruction algorithm}\label{Sec:Heat_Reconstruction_algorithm}

This section presents the implementation and validation of the reconstruction procedure,  summarized in Algorithm~\ref{alg:Framwork_of_Heateq}. Recall that the graph $\mathbb{G}=(X,\mathcal{E},\mu,w)$ satisfies Assumption \ref{Assum_eigenfunction}.

\begin{algorithm}[htb]
\caption{Reconstruction Algorithm for $\mu|_{X\setminus B}$.}
\label{alg:Framwork_of_Heateq}
\begin{algorithmic}[1] 
\REQUIRE $\mathbb{G} = (X,\mathcal{E},\mu,w)$ that satisfies the assumptions in Corollary \ref{coro_dense_for_Heatweight}; the distribution of the first passage times $r|_{\mathbb{Z}_{2T}\setminus \{0\} \times B \times B}$.

    \STATE Compute $U^f|_{\mathbb{Z}_{2T}\times B}$ from the distributions of the first passage times via
\small{\begin{align*}
U^f(t,x) =
\begin{cases}
0, \quad & t=0, \\
f(0,x),\quad &t=1,\\
\sum\limits_{y\in B}f(t-2,y)\cdot r(1,x,y)+\sum\limits_{s=1}^{t-2}\sum\limits_{y\in B} \bigg\{f(s-1,y)\Big( r(t-s,x,y)\\
\quad +\sum\limits_{j=2}^{t-s}\sum\limits_{1\leq t_1< t_2<\cdots<t_j=t-s} r(t_1,x,y)\cdot  \prod\limits_{i=2}^jr(t_i-t_{i-1},y,y) \Big) \bigg\} + f(t-1,x),\quad & t\geq 2
\end{cases}
\end{align*}}
for $(t,x)\in \mathbb{Z}_{2T}\times B$  (see \eqref{Uf_by_combinations_by_r}). 

    \STATE Compute the operator $W^*W$ by varying $f$ in $W^*W{f}(t,x)=P_T R_{2T-1}U^{f}$ for $(t,x)\in\mathbb{Z}_T\times B$, see \eqref{Heat_Inner_product_of_U^f}.
    
    \STATE Compute the functions $W^* {\varphi}^{(j)}$ ($j=1,\cdots, |B|$) by
    \begin{align*}
    W^*\varphi^{(j)}=R_{T-1} \Lambda_{\mu} R_{T-1} P_B \Delta \varphi^{(j)}+P_B \varphi^{(j)},
    \end{align*}
    see \eqref{Heat_Wstar_formula} and \eqref{Dirichlet_problem}.

    \STATE Compute $h^{(j)}_0 = (W^*W)^\dagger W^* \varphi^{(j)}$ for each  $j=1,\cdots, |B|$,  see \eqref{Heat_B_h0formula}. 
    
       \STATE Solve the linear equations with $1\leq j\le k\leq |B|$:
\small{\begin{align*}
\sum_{x\in X\setminus B}\mu_x\varphi^{(j)}(x)\varphi^{(k)}(x)= \langle h^{(j)}_0,W^*\varphi^{(k)}\rangle_{\mathbb{Z}_T\times B}-\sum_{x\in B}\mu_x\varphi^{(j)}(x)\varphi^{(k)}(x),
\end{align*}}
see \eqref{Heat_miureconstruction_formula}.

\RETURN $\mu|_{X\setminus B}$.

\vspace{1ex}
\ENSURE The vertex centrality $\mu|_{X\setminus B}$. 

\end{algorithmic}
\end{algorithm}

We introduce a vectorization process that converts functions on graphs into vectors for the implementation of Algorithm \ref{alg:Framwork_of_Heateq}.
For a time-independent function $g\in\ell^2(X)$, we continue to use the vectorization described in \eqref{eq:Heat_identity}.

For a time-dependent function $u\in \ell^2(\mathbb{Z}_T\times X)$, we use the lexicographical order to identify
\begin{align*}
u \leftrightarrow \vec{u} := & (u(0,x_1), u(0,x_2), \cdots, u(0,x_{|X|}), u(1,x_1), u(1,x_2),\cdots, u(1,x_{|X|}), \dots, \\
 & u(T-1,x_1), u(T-1,x_2), \dots, u(T-1,x_{|X|}))^\top.
\end{align*}
Similarly, we can vectorize a function in
$\ell^2(\mathbb{Z}_T\times (X\setminus B))$ or $\ell^2(\mathbb{Z}_T\times B)$ following the ordering of vertices in $X$ as above.
Sometimes, we need to extend $g\in\ell^2(B)$ to a function in $\ell^2(\mathbb{Z}_T\times B)$ that is constant in the variable $t$. Following the lexicographical order, such extension is vectorized as $\mathbf{1}_T\otimes \vec{g}$, where $\mathbf{1}_T := (1,\dots,1)^\top \in\mathbb{R}^T$ denotes the vector of all 1's, and $\otimes$ is the vector tensor product. Based on this ordering of the vertices, linear operators in the algorithm are realized as matrices. We put a square parenthesis $[\cdot]$ around the linear operator to denote its matrix realization.

Implementation of Algorithm~\ref{alg:Framwork_of_Heateq} consists of the following steps:

{\bf{Step 1: Assemble 
$\overrightarrow{{U^{f}}}$   and the matrix $[\Lambda_{\mu}]$.}} Given any $f\in\ell^2(\mathbb{Z}_T\times B)$ with the vectorization $\vec{f}\in \mathbb{R}^{T |B|}$ and the data $r|_{\mathbb{Z}_{2T}\setminus\{0\} \times B \times B}$, we use the formula \eqref{Uf_by_combinations_by_r} to obtain $U^f|_{\mathbb{Z}_{2T}\times B}$ and its vectorization $\overrightarrow{{U^{f}}}\in\mathbb{R}^{2T|B|}$. 
Since the matrix $[\Lambda_{\mu}]\in\mathbb{R}^{T|B|\times T|B|}$ satisfies
$[\Lambda_{\mu}] \vec{f} = [P_T] \overrightarrow{{U^{f}}}$ by \eqref{Map_of_source to solution_of_Nheat}, ranging $\vec{f}$ over all the standard basis vectors in $\mathbb{R}^{T |B|}$ gives $[\Lambda_{\mu}]$.

{\bf{Step 2: Construct the matrix $[W^{*}W]$.}} 
For any $f\in\ell^2(\mathbb{Z}_T\times B)$ with vectorization $\vec{f}\in \mathbb{R}^{T |B|}$, the matrix $[W^{*}W]\in \mathbb{R}^{T |B|\times  T|B|}$ is determined by the following relation in Proposition~\ref{prop:product_of_Uf12}:
\begin{align}\label{Heat_WstarW}
[W^*W]\vec{f} = [P_T] [R_{2T-1}] \overrightarrow{{U^{f}}}
\end{align}
where $[P_T]\in \mathbb{R}^{T|B|\times 2T|B|}$ and $[R_{2T-1}]\in \mathbb{R}^{2T|B|\times 2T|B|}$. Ranging $\vec{f}$ over all the standard basis vectors in $\mathbb{R}^{T |B|}$ gives the matrix $[W^*W]$.

{\bf{Step 3: Compute the matrix-vector product $[W^*] \vec{\varphi}^{(j)}$ for $j=1,2,\dots,|B|$.}} For each $j$, the solution $\varphi^{(j)}$ of \eqref{Dirichlet_problem} is vectorized as $\vec{\varphi}^{(j)}\in\mathbb{R}^{|X|}$. We use \eqref{Heat_Wstar_formula} to get
$$
[W^*] \vec{\varphi}^{(j)} = [R_{T-1}] [\Lambda_{\mu}] [R_{T-1}] \left( \mathbf{1}_T \otimes [P_B] [\Delta] \vec{\varphi}^{(j)} \right) + \mathbf{1}_T \otimes  [P_B] \vec{\varphi}^{(j)}.
$$
Here, $[R_{T-1}]\in\mathbb{R}^{T|B|\times T|B|}$, $[\Lambda_{\mu}]\in\mathbb{R}^{T|B|\times T|B|}$, $[P_B]\in\mathbb{R}^{|B|\times |X|}$, $[\Delta]\in\mathbb{R}^{|X|\times |X|}$. The tensor product $\mathbf{1}_T \otimes \cdot$ extends a vectorized spatial function to a vectorized spatial-temporal function.

{\bf{Step 4: Determine the control function $\vec{h}_0^{(j)}$ for each $\vec{\varphi}^{(j)}$.}} Using Proposition~\ref{Heat_B_harmonic_explicit h0}, we compute
\begin{align}\label{Heat_h0}
\vec{h}^{(j)}_0 := [W^* W]^\dagger [W^*] \vec{\varphi}^{(j)} \qquad \in \mathbb{R}^{T |B|}
\end{align}
for each $j=1,2,\dots,|B|$. This is possible as the matrix $[W^* W]$ and the vector $[W^*] \vec{\varphi}^{(j)}$ have been obtained in the previous steps.

{\bf{Step 5: Reconstruct  $\vec{\mu}_{X\setminus B}$.}}
For $1\le j \le k \le |B|$, we solve the following linear system arising from \eqref{Heat_miureconstruction_formula}:
\begin{equation}\label{Heat_mu}
 (\vec{\varphi}^{(j)}_{X\setminus B} \odot \vec{\varphi}^{(k)}_{X\setminus B})^\top \vec{\mu}_{X\setminus B}= \vec{h}^{ (j)\top}_0 [W^*] \vec{\varphi}^{(k)} - (\vec{\varphi}^{(j)}_B \odot \vec{\varphi}^{(k)}_B)^\top\vec{\mu}_B, \qquad 1\leq j\leq k\leq |B|.
\end{equation}
Note that the coefficient matrix of this linear system is $\mathbf{H}^\top$. If this equation admits a unique solution, we can recover $\vec{\mu}_{X\setminus B}$. Otherwise, the minimum norm solution gives the orthogonal projection of $\vec{\mu}_{X\setminus B}$ onto $Q$.

\section{Numerical experiment}\label{Sec:Heat_Numerical_experiment}

In this section, we validate the algorithm on several numerical examples using MATLAB and Python.
The experimental equipment used is a laptop with the following hardware specifications: CPU: Intel(R) Core(TM) i7-10510U @1.80GHz and RAM: 16.0 GB.

Given a  graph $\mathbb{G}=(X,\mathcal{E},\mu,w)$ with known vertex centrality and edge weight, we first simulate an \textit{empirical distribution of the first passage times}, denoted by 
$$
r'(t,x,y), \qquad (t,x,y) \in\mathbb{Z}_{2T}\setminus\{0\}\times B \times B.
$$ 
The simulation is done by the Monte Carlo method, and the empirical distribution is used later as the approximate measurement data to validate Algorithm~\ref{alg:Framwork_of_Heateq}. 
For a fixed pair of vertices \( x,y \in B \), we first simulate the function $\tau(x,y)$. Each simulation sets up a random walk that begins at $x$, then randomly selects the next adjacent vertex to be visited according to the transition probabilities~\eqref{Heat_transition_proba_def}. If $y$ is visited, the current step count is recorded as $\tau(x,y)$ and the simulation is terminated. If $y$ is never visited within time $2T-1$, $\tau(x,y)$ is set to infinity. Therefore, the possible values for $\tau(x,y)$ are $1,2,\dots,2T-1$ and $\infty$. We repeat such simulation $10^6$ times and record the histogram of the possible values. The histogram with $t<\infty$ gives $r'(t,x,y)$ for fixed $x,y\in B$. Repeating this process for all $x,y\in B$ yields $r'(t,x,y)$ with $(t,x,y) \in\mathbb{Z}_{2T}\setminus\{0\}\times B \times B$, which is used as approximate data in the following experiments.

We would like to evaluate the discrepancy of the empirical distribution of the first passage times $r'|_{\mathbb{Z}_{2T}\setminus\{0\}\times B \times B}$ related to the genuine distribution $r|_{\mathbb{Z}_{2T}\setminus\{0\}\times B \times B}$. For this purpose, given a ground-truth vertex centrality $\mu$ on $X$ and the edge weight $w$ on $\mathcal{E}$, we compute the transition probabilities $p_{xy}$ by~\eqref{Heat_transition_proba_def}, then compute $r$ based on the following recursive relation (e.g, see \cite[Chapter 4.8]{IBE201359} for an account): 
\begin{align*}
r(t,x,y)=
\begin{cases}
 p_{xy},\quad & t=1,\\
\sum\limits_{X\ni z\neq y}p_{xz}\cdot r(t-1,z,y),\quad & t = 2,3,\dots
\end{cases} \qquad x,y\in X.
\end{align*}
The discretized empirical and genuine distributions of the first passage times are denoted by $[r'],[r]\in \mathbb{R}^{(2T-1)\times|B|\times|B|}$, respectively. 
We will use the Frobenius relative norm error ($\mathrm {FRNE}$)
$$
\mathrm {FRNE}=\frac{\lVert [r]-[r']  \rVert_{F} }{\lVert [r]\rVert_{F} }*100\%
$$
as a metric to quantify the relative discrepancy between the empirical and genuine distributions.

On the other hand, denote the discretized ground-truth vertex centrality and the reconstruction by $\vec{\mu}_x$ and $\vec{\mu}_x'$, respectively.
We will use both the absolute error 
\[\mathrm {Error}:=|\vec{\mu}_x -\vec{\mu}_x'|\]
and the $L_2$-relative norm error ($\mathrm {L_2RNE}$)
\[\mathrm {L_2RNE}:=\frac{\lVert \vec{\mu}_x-\vec{\mu}_x'\rVert_{2} }{\lVert\vec{\mu}_x \rVert_{2} }* 100\%\]
to quantify the accuracy of the reconstruction.

In the following, Algorithm~\ref{alg:Framwork_of_Heateq} is validated using numerical experiments. The graphs in these experiments are limited in size primarily due to the algorithmic complexity, which grows exponentially in the number of vertices $|X|$, see Remark~\ref{Operation_complexity_of_Uf}. Another limiting factor for the size of the graphs comes from the Monte Carlo simulation to generate $r'$. It is well known that the Monte Carlo simulation converges at the rate of $O(N^{-\frac{1}{2}})$, where $N$ is the number of samples. This slow rate of convergence requires a significant number of random simulations to obtain a reasonably accurate $r'$ as the approximate data. In our experiments, we run the simulation $10^6$ times for each fixed pair of $x,y\in B$ in order to gain three significant digits of accuracy.

\subsection{Experiment 1: a graph with eight vertices}

This experiment tests Algorithm~\ref{alg:Framwork_of_Heateq} when the ground-truth $\mu$ is a constant. We choose $|B|=3, |X\setminus B|=5$, $w_{x,y}=0.25, T=9$, and the ground-truth vertex centrality is $\mu_x= 1$ for all $x \in X$. Note that the choice of these values satisfies the assumption in Corollary~\ref{coro_dense_for_Heatweight}. 
The graph is shown in Fig. \ref{fig_8vertex_graph}.
The empirical and genuine distributions are computed as explained above. The FRNE between $[r']$ and $[r]$ is $0.248499\%$.

\begin{figure}[htbp]
\centering

\includegraphics[scale=0.5]{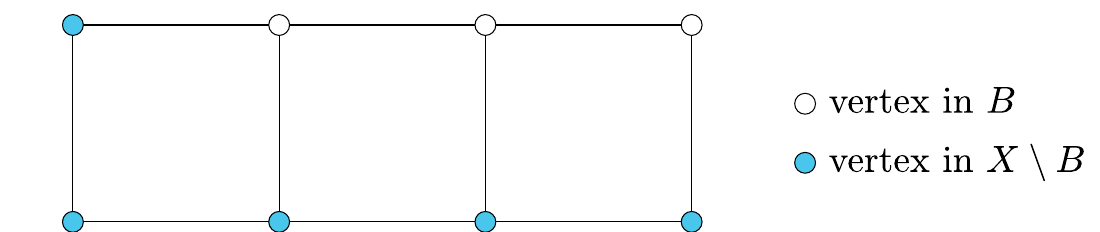}
\caption{\small\small{$X$ and $B$ in Experiment 1.}}
\label{fig_8vertex_graph}
\end{figure}

Following Algorithm \ref{alg:Framwork_of_Heateq}, we compute the vector $\overrightarrow{{U^{f}}} \in\mathbb{R}^{54}$ and the matrix $[\Lambda_{\mu}]\in\mathbb{R}^{27\times 27}$ by \eqref{Uf_by_combinations_by_r}, the matrix $[W^{*} W]\in \mathbb{R}^{27\times 27}$ by \eqref{Heat_h0}, and the matrix $\mathbf{H}\in \mathbb{R}^{5\times 6}$ by \eqref{Heat_mu}. 
In particular, we observe that $[W^{*} W]$ and $\mathbf{H}$ are ill-conditioned. Their singular values are shown in Fig. \ref{fig_Heat_8vetertex_singular}.
The ill-conditioning indicates that suitable regularization is needed when solving the least-squares problem \eqref{Heat_h0} for $\vec{h}_0$ and the least-squares problem \eqref{Heat_mu} for $\vec{\mu}|_{X\setminus B}$. Here, we utilize the rank-revealing QR method (with column pivoting) to solve these least-squares problems. Specifically, given a matrix $A=[W^{*} W] \text{ or } \mathbf{H}^\top$, this method performs a column-pivoted QR decomposition on $A$ and zeros out the diagonal entries in the triangular matrix that are under a certain threshold ``\textit{tol}'', reducing the problem to a simpler triangular form that can be solved for the minimum-norm solution. 
In the experiments, this method is implemented using the MATLAB command `\textit{lsqminnorm}' with the threshold `\textit{tol = 0.0005}'.
The reconstruction and the error of vertex centrality are shown in Fig. \ref{fig_Heat_8vetertex_miu_rec}. The relative reconstruction error is $\mathrm {L_2RNE}=0.507423\%$.

\begin{figure}[htbp]
\centering

\includegraphics[scale=0.56]{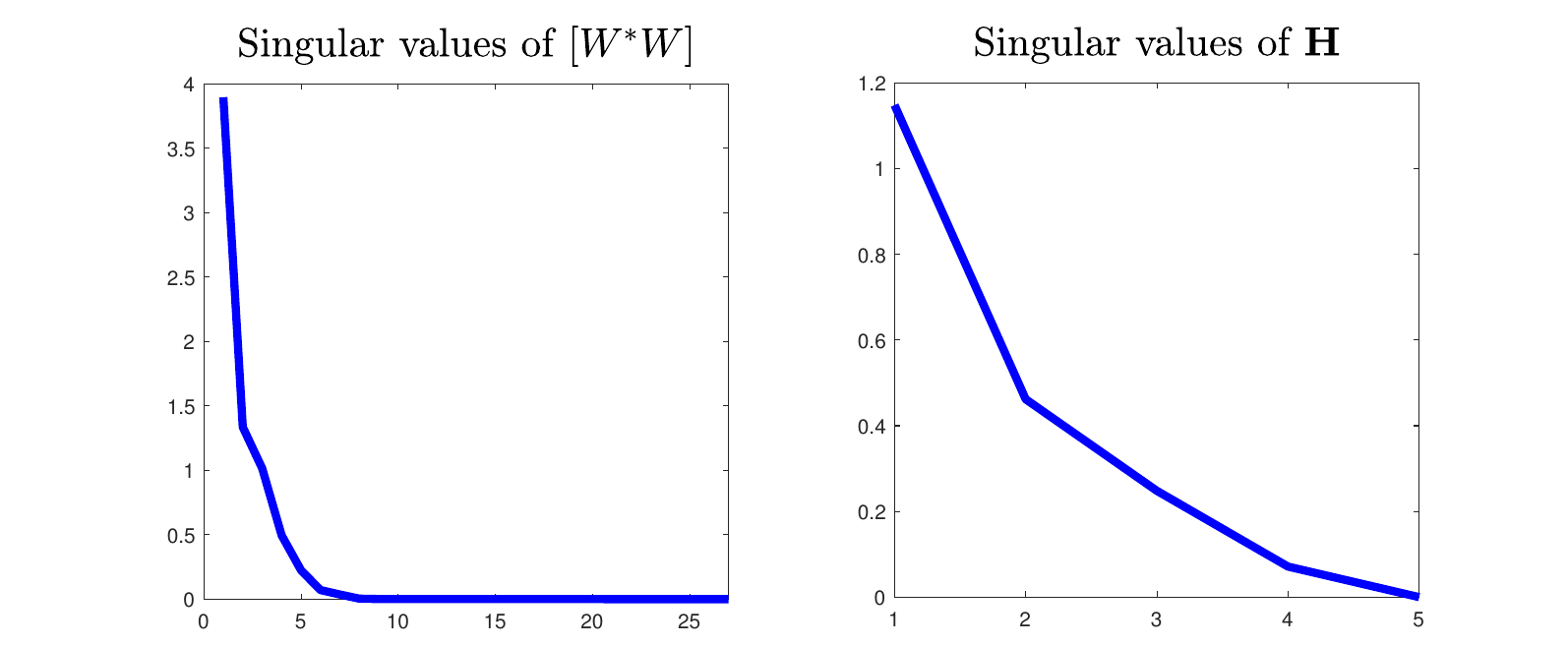}
\caption{\small\small{The singular values of $[W^{*}W]$ and $\mathbf{H}$ in Experiment 1. The minimum singular values are $1.6083*10^{-5}$ and  $1.5844*10^{-17}$, respectively.}}
\label{fig_Heat_8vetertex_singular}
\end{figure}

\begin{figure}[htbp]
\centering

\includegraphics[scale=0.5]{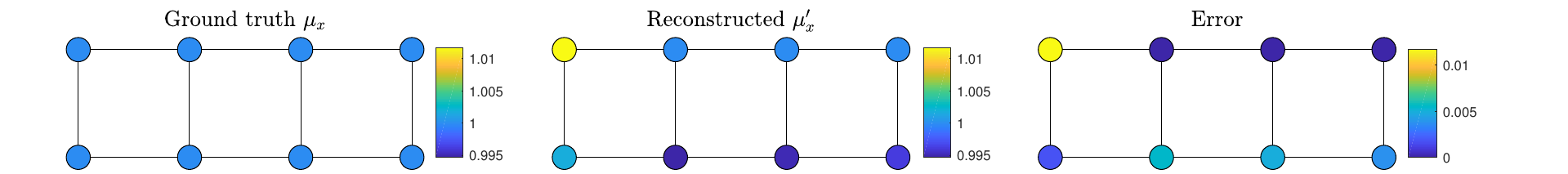}
\caption{\small\small{The ground-truth $\mu_x$, the reconstructed $\mu_x'$ and the absolute errors  in Experiment 1.}}
\label{fig_Heat_8vetertex_miu_rec}
\end{figure}

\subsection{Experiment 2: the graph with nine vertices}

This experiment tests Algorithm~\ref{alg:Framwork_of_Heateq} when the ground-truth $\mu$ is varying. We choose $|B|=3$, $|X\setminus B|=6$, $w_{x,y}=1, T=9$, and the ground-truth vertex centrality is $\mu_x= {\rm {deg}} (x)$ for all $x \in X$, where the degree ${\rm {deg}} (x)$ is defined as the number of edges connected to the vertex $x$. Note that the choice of these values satisfies the assumption in Corollary~\ref{coro_dense_for_Heatweight}. 
The graph is shown in Fig. \ref{fig_9vertex_graph}.
The FRNE between $[r']$ and $[r]$ is $0.188348\%$.

\begin{figure}[htbp]
\centering

\includegraphics[scale=0.52]{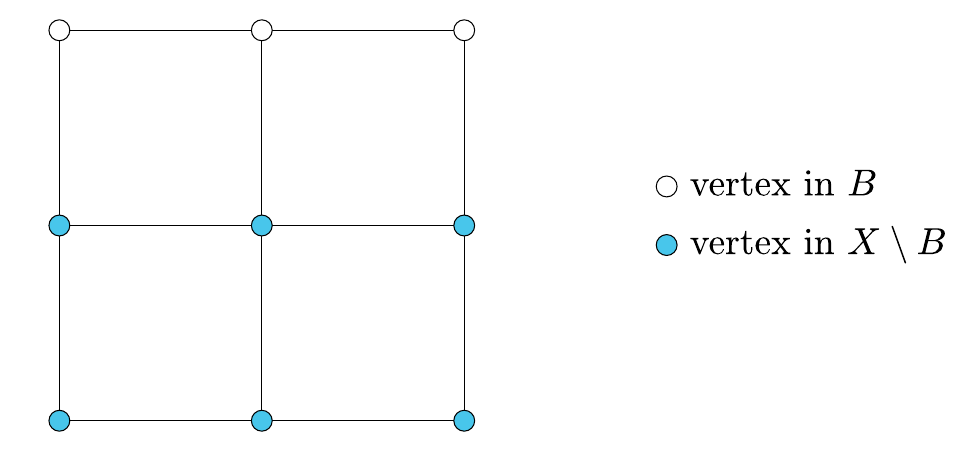}
\caption{\small\small{$X$ and $B$ in Experiment 2.}}
\label{fig_9vertex_graph}
\end{figure}

This time, we have $\overrightarrow{{U^{f}}} \in\mathbb{R}^{54}$, $[\Lambda_{\mu}]\in\mathbb{R}^{27\times 27}$ by \eqref{Uf_by_combinations_by_r}, $[W^{*} W]\in \mathbb{R}^{27\times 27}$ by \eqref{Heat_h0}, and $\mathbf{H}\in \mathbb{R}^{6\times 6}$ by \eqref{Heat_mu}. The matrices $[W^{*} W]$ and $\mathbf{H}$ remain slightly ill-conditioned. Their singular values are shown in Fig. \ref{fig_Heat_9vetertex_singular}.
The rank-revealing QR method (with column pivoting) is employed again to solve the least-squares problems for $\vec{h}_0$  with the threshold `\textit{tol = 0.0012}' and $\vec{\mu}|_{X\setminus B}$, respectively. The reconstruction and the error of vertex centrality are shown in Fig. \ref{fig_Heat_9vetertex_miu_rec}. The relative reconstruction error is $\mathrm {L_2RNE}=4.570132\%$.

\begin{figure}[htbp]
\centering
\includegraphics[scale=0.56]{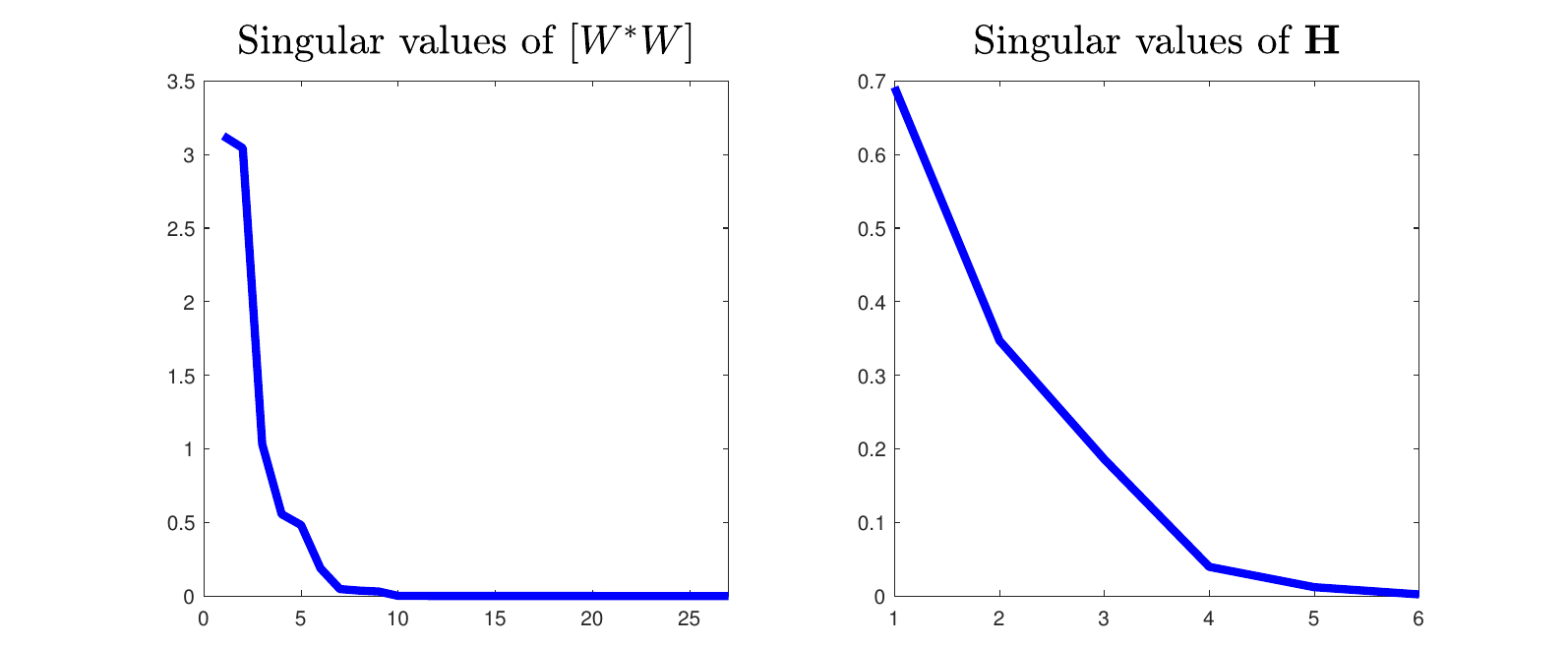}
\caption{\small\small{The singular values of $[W^{*}W]$ and $\mathbf{H}$ in Experiment 2. The minimum singular values are $1.0922*10^{-4}$ and $0.0023$, respectively.}}
\label{fig_Heat_9vetertex_singular}
\end{figure}

\begin{figure}[htbp]
\centering
\includegraphics[scale=0.6]{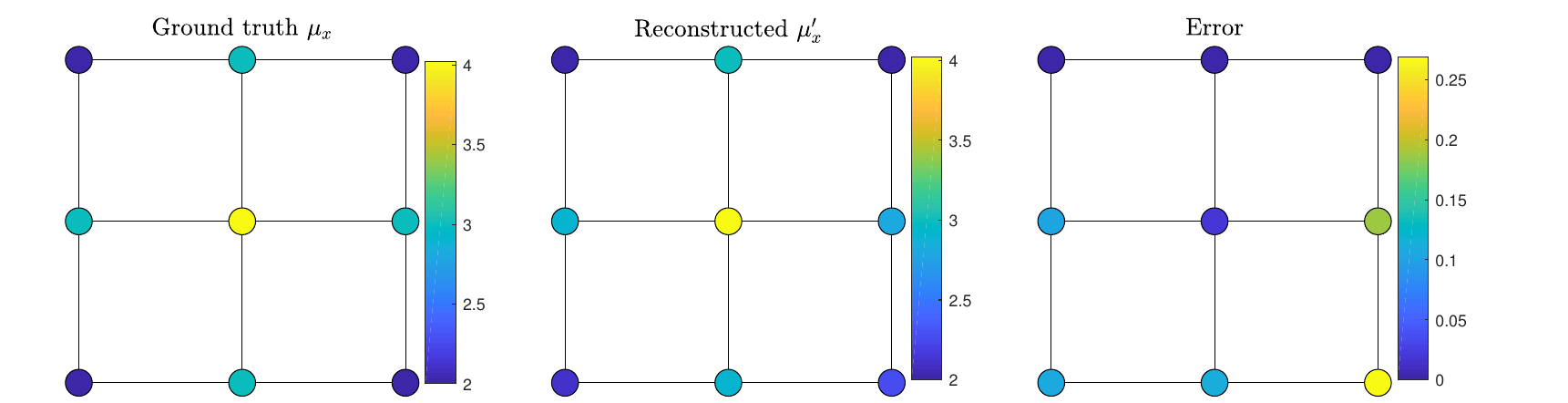}
\caption{\small\small{The ground-truth $\mu_x$, the reconstructed $\mu_x'$ and the absolute errors in Experiment 2.
}}
\label{fig_Heat_9vetertex_miu_rec}
\end{figure}

\begin{appendix}

\section{}\label{Sec:aAppSourceSolutionAdjiont}

In this appendix, we characterize the adjoint of the operator $\Lambda_{\mu}$. For the subset $B \subset X$, the adjoint $\Lambda_{\mu}^*$ acts on functions defined on the space-time domain $\mathbb{Z}_T \times B$.  We demonstrate that its action coincides with the time-reversal operator and  $\Lambda_{\mu}$ itself. To establish this, we define the backward time difference operator $D_t^*$ by
\[ D_t^*v(t,\cdot):= v(t-1,\cdot)-v(t,\cdot)\quad \text{for } ~t\geq 1.\]
Let $U^f$ be the solution to the heat equation \eqref{nonhomog_graph_heat_eq}.
Now consider a function $v$ satisfying 
the following terminal-value problem
\begin{align}\label{Heat_calculateLambdamustar}
\begin{cases}
D_t^* v(t,x) - \Delta v(t,x) = g(t,x), & (t,x)\in \mathbb{Z}_T\setminus\{0\}\times X, \\
v(T-1,x) = 0, & x\in X,
\end{cases}
\end{align}
where the source  $g$ has support contained in $\mathbb{Z}_T\times B$.

Because of the definition $\Lambda_{\mu} f = U^f|_{\mathbb{Z}_T\times B}$,
and $\operatorname{supp}(g) \subset \mathbb{Z}_T \times B$, we extend the domain of integration $\langle g,U^f \rangle$ from $B$ to the entire  set of vertices $X$ without changing the value when $t\in \mathbb{Z}_T$ in the following.
Recall $v$, we have
\begin{align}\label{Al1}
\langle g,\Lambda_{\mu}f\rangle_{\mathbb{Z}_T\times B}
& = \langle g,U^f \rangle_{\mathbb{Z}_T\times B} \nonumber\\
& =  \langle g, U^f \rangle_{\mathbb{Z}_T\times X} \nonumber\\
& = \langle g(0,\cdot),U^f(0,\cdot)\rangle_X + \langle D_t^* v(t,x) - \Delta v(t,x), U^f(t,x)\rangle_{\{1,2,\dots, T-1\}\times X}\nonumber\\
& = \langle v(t-1,x)- v(t,x) - \Delta v(t,x), U^f(t,x)\rangle_{\{1,2,\dots, T-1\}\times X}\nonumber\\
& = \langle v(t-1,x),U^f(t,x))_{\{1,2,\dots, T-1\}\times X}- (v(t,x),U^f(t,x)\rangle_{\{1,2,\dots, T-1\}\times X}\nonumber\\
&\quad - (\Delta v(t,x), U^f(t,x))_{\{1,2,\dots, T-1\}\times X}\nonumber\\
& = \langle v(t,x),U^f(t+1,x))_{\{0,1,\dots, T-2\}\times X}- (v(t,x),U^f(t,x)\rangle_{\{1,2,\dots, T-1\}\times X}\nonumber\\
&\quad - \langle v(t,x), \Delta  U^f(t,x)\rangle_{\{1,2,\dots, T-1\}\times X}\nonumber\\
& = \langle v(t,x),U^f(t+1,x)-U^f(t,x)-\Delta U^f(t,x)\rangle_{\{0,1,\dots, T-1\}\times X}\nonumber\\
&\quad - \langle v(T-1,x), U^f(T,x)\rangle_X\nonumber\\
& = \langle v,f\rangle_{\mathbb{Z}_T\times X}\nonumber\\
& = \langle v,f\rangle_{\mathbb{Z}_T\times B},
\end{align}
where the self-adjoint of $\Delta$ on $X$ is used.

To express this adjoint relationship in explicit form, we recall the time reversal operator $R_{T-1}$, defined by $R_{T-1} u(t, \cdot) = u(T-1-t, \cdot)$ on the discrete time interval $\mathbb{Z}_T$.
Now consider the following initial-value problem
\begin{align}\label{Heat_Timerever_calculateLambdamustar}
\begin{cases}
D_t V(t,x) - \Delta V(t,x) = R_{T-1}g(t,x), & (t,x)\in \mathbb{Z}_T\times X, \\
V(0,x) = 0, & x\in X,
\end{cases}
\end{align}
where the source  $g$ satisfies $\operatorname{supp}(g) \subset \mathbb{Z}_T\times B$.

Let $v$ be the solution to the backward problem \eqref{Heat_calculateLambdamustar}. By comparing the definitions, one observes that $v(t,x) = R_{T-1} V(t,x)$ for all $(t,x)$. Indeed, applying $R_{T-1}$ to the forward problem \eqref{Heat_Timerever_calculateLambdamustar} yields the backward problem for $v$. Consequently, we have 
\begin{align*}
R_{T-1}(\Lambda_{\mu}(R_{T-1}g))=R_{T-1}V|_{\mathbb{Z}_T\times B}=v|_{\mathbb{Z}_T\times B}.
\end{align*}
Substituting this relation into the inner product identity \eqref{Al1} yields
\begin{align*}
\langle g, \Lambda_{\mu} f \rangle_{\{0,1,\dots, T-1\}\times B} 
&= \langle v, f \rangle_{\mathbb{Z}_T\times B} \\
&= \langle R_{T-1}\Lambda_{\mu}R_{T-1}g, f \rangle_{\mathbb{Z}_T\times B}.
\end{align*}
Since this holds for all $g$ and $f$ with support in $\mathbb{Z}_T\times B$, it follows that the equality holds for all test functions $f$ and $g$ with support in $\mathbb{Z}_T\times B$, we conclude that the adjoint of $\Lambda_{\mu}$ on this space is given by
\[
\Lambda_{\mu}^* = R_{T-1} \Lambda_{\mu} R_{T-1}.\]

\end{appendix}

\bibliographystyle{abbrv}

\bibliography{GYL}

\end{document}